\documentclass{llncs}
\usepackage{amssymb,amsmath}
\usepackage{graphicx,color}
\usepackage{wasysym}
\begin{document}
\frontmatter          
\pagestyle{headings}  

\mainmatter   
%
%
\title{Fixed-Orientation Equilateral Triangle Matching of Point Sets \thanks{Anil Maheshwari and Michiel Smid are supported by NSERC. Jasine Babu is supported by DFAIT.}}
%

%

\titlerunning{Fixed-Orientation Equilateral Triangle Matching of point sets}  
\author{Jasine Babu \inst{1} \and Ahmad Biniaz \inst{2} \and Anil Maheshwari \inst{2} \and Michiel Smid \inst{2}}
\authorrunning{Babu et al.} 
\institute{Department of Computer Science and Automation,\\Indian Institute of Science, Bangalore, India.\\
\and School of Computer Science, Carleton University, Ottawa, Canada
\email{jasine@csa.iisc.ernet.in, ahmad.biniaz@gmail.com, \{anil, michiel\}@scs.carleton.ca}}
\maketitle  

%
%
 


\begin{abstract}
Given a point set $P$ and a class $\mathcal{C}$ of geometric objects, $G_\mathcal{C}(P)$ is a geometric graph with vertex set $P$ such that any two vertices $p$ and $q$ are adjacent if and only if there is some $C \in \mathcal{C}$ containing both $p$ and $q$ but no other points from $P$. We study $G_{\bigtriangledown}(P)$ graphs where $\bigtriangledown$ is the class of downward equilateral triangles (ie. equilateral triangles with one of their sides parallel to the $x$-axis and the corner opposite to this side below that side). For point sets in general position, these graphs have been shown to be equivalent to half-$\Theta_6$ graphs and TD-Delaunay graphs. 

The main result in our paper is that for point sets $P$ in general position, $G_{\bigtriangledown}(P)$ always contains a matching of size at least $\left\lceil\frac{n-2}{3}\right\rceil$ and this bound cannot be improved above $\left\lceil\frac{n-1}{3}\right\rceil$.

We also give some structural properties of $G_{\davidsstar}(P)$ graphs, where $\davidsstar$ is the class which contains both upward and downward equilateral triangles. We show that for point sets in general position, the block cut point graph of $G_{\davidsstar}(P)$ is simply a path. Through the equivalence of $G_{\davidsstar}(P)$ graphs with $\Theta_6$ graphs, we also derive that any $\Theta_6$ graph can have at most $5n-11$ edges, for point sets in general position.
\end{abstract} 
 \section{Introduction}
In this work, we study the structural properties of some special geometric graphs defined on a set $P$  of $n$ points on the plane. An equilateral triangle with one side parallel to the $x$-axis and the corner opposite to this side below
(resp. above) that side as in $\bigtriangledown$ (resp. $\bigtriangleup$) will be called a down (resp. up)-triangle.
A point set $P$ is said to be in general position, if the line passing through any two points from $P$ does not make angles $0^\circ$, $60 ^\circ$ or $120^\circ$ with the horizontal \cite{Bonichon2010,Panahi}. In this paper, we consider only point sets that are in general position and our results assume this pre-condition.  

Given a point set $P$, $G_{\bigtriangledown}(P)$ (resp. $G_{\bigtriangleup}(P)$) is defined as the graph whose vertex set is $P$ and that has an edge between any two vertices $p$ and $q$ if and only if there is a down-(resp. up-)triangle containing both points $p$ and $q$ but no other points from $P$ (See Fig. \ref{graph}.).
We also define another graph $G_{\davidsstar}(P)$ as the graph whose vertex set is $P$ and that has an edge between any two vertices $p$ and $q$ if and only if there is a down-triangle or an up-triangle containing both points $p$ and $q$ but no other points from $P$. In Section \ref{prelims} we will see that, for any point set $P$ in general position, its $G_{\bigtriangledown}(P)$ graph is the same as the well known Triangle Distance Delaunay (TD-Delaunay) graph of $P$ and the half-$\Theta_6$ graph of $P$ on so-called negative cones. Moreover, $G_{\davidsstar}(P)$ is the same as the $\Theta_6$ graph of $P$ \cite{Bonichon2010,Chew1989}. 
\begin{figure}
\centering
  \includegraphics[scale=0.5]{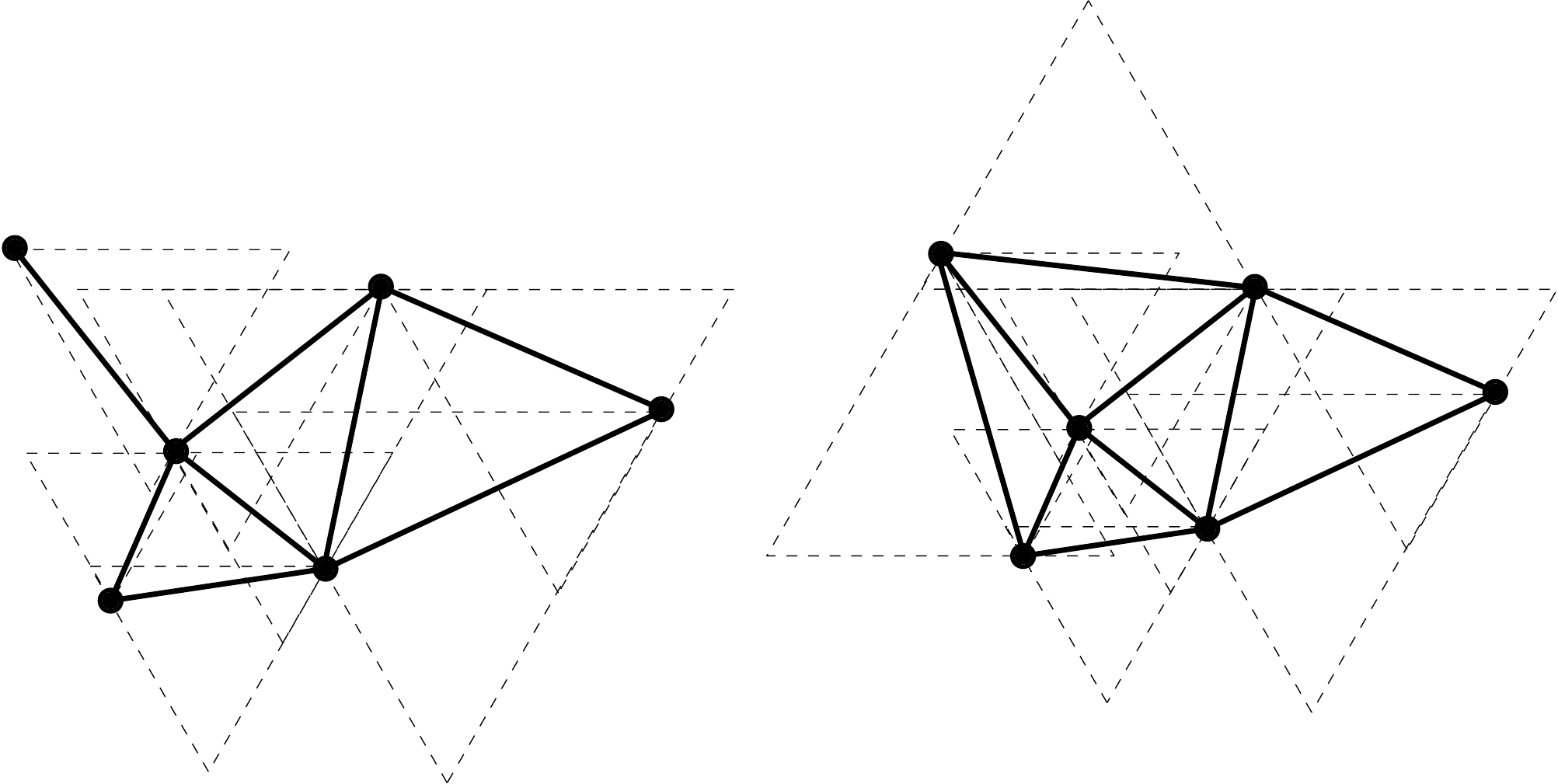}   
  \caption{A point set $P$ and its (a) $G_{\bigtriangledown}(P)$ and (b) $G_{\davidsstar}(P)$.}
\label{graph}
  \end{figure}

Given a point set $P$ and a class $\mathcal{C}$ of geometric objects, the maximum $\mathcal{C}$-matching problem is to compute a subclass $\mathcal{C}'$ of $\mathcal{C}$ of maximum cardinality such that no point from $P$ belongs to more than one element of $\mathcal{C}'$ and for each $C \in \mathcal{C}'$, there are exactly two points from $P$ which lie inside $C$. Dillencourt \cite{Dillencourt1990} proved that every point set admits a perfect circle-matching. \'{A}brego et al. \cite{Abrego2009} studied the isothetic square matching  problem. Bereg et al. concentrated on matching points using axis-aligned squares and rectangles \cite{Bereg2009}.

A matching in a graph $G$ is a subset $M$ of the edge set of $G$ such that no two edges in $M$ share a common end-point. A matching is called a maximum matching if its cardinality is the maximum among all possible matchings in $G$. If all vertices of $G$ appear as end-points of some edge in the matching, then it is called a perfect matching. It is not difficult to see that for a class $\mathcal{C}$ of geometric objects, computing the maximum $\mathcal{C}$-matching of a point set $P$ is equivalent to computing the maximum matching in the graph $G_\mathcal{C}(P)$.

The maximum $\bigtriangleup$-matching problem, which is the same as the maximum matching problem on $G_\bigtriangleup(P)$, was previously studied by Panahi et al. \cite{Panahi}. It was claimed that, for any point set $P$ of $n$ points in general position, any maximum matching of $G_\bigtriangleup(P)$ (and $G_{\bigtriangledown}(P)$) will match at least $\lfloor \frac{2n}{3} \rfloor$ vertices. But we found that their proof of Lemma 7, which is very crucial for their result, has gaps. By a completely different approach, we show that for any point set $P$ in general position, $G_\bigtriangledown(P)$ (and by symmetric arguments, $G_{\bigtriangleup}(P)$) will have a maximum matching of size at least $\lceil\frac{n-2}{3} \rceil$; i.e, at least $2(\lceil\frac{n-2}{3} \rceil)$ vertices are matched. We also give examples where our bound is tight, in all cases except when $|P|$ is one less than a multiple of three.

We also prove some structural and geometric properties of the graphs $G_{\bigtriangledown}(P)$ (and by symmetric arguments, $G_{\bigtriangleup}(P)$) and $G_{\davidsstar}(P)$. It will follow that for point sets in general position, $\Theta_6$ graphs can have at most $5n-11$ edges and their block cut point graph is a simple path. 
\section{Notations}
Our notations are similar to those used in \cite{Bonichon2010}, with some minor modifications adopted for convenience. 
A {\em cone} is the region in the plane between two rays that emanate from the same point, its apex. 
Consider the rays obtained by a counter-clockwise rotation of the positive $x$-axis by angles of $\frac{i\pi}{3}$ with $i=1, \dots, 6$ 
around a point $p$. (See Fig. \ref{Figcones}.) 
Each pair of successive rays, $\frac{(i-1)\pi}{3}$ and $\frac{i\pi}{3}$, defines a cone, denoted by $A_i(p)$, whose apex is $p$. 
For $i \in \{1, \ldots, 6\}$, when $i$ is odd, we denote $A_i(p)$ using $C_{\frac{i+1}{2}}(p)$ and the cone opposite to $C_i(p)$ using  $\overline{C}_i(p)$. We call $C_i(p)$ a positive cone around $p$ and $\overline{C_i}(p)$ a negative cone around $p$. For each cone $\overline{C_i}(p)$ (resp. $C_i(p)$), let $\ell_{\overline{C_i}(p)}$ (resp. $\ell_{{C_i}(p)}$) be its bisector. If $p' \in \overline{C_i}(p)$, then let $\overline{c_i}(p, p')$ denote the distance between $p$ and the orthogonal projection of $p'$ onto $\ell_{\overline{C_i}(p)}$. Similarly, if $p' \in C_i(p)$, then let ${c_i}(p, p')$ denote the distance between $p$ and the orthogonal projection of $p'$ onto $\ell_{C_i(p)}$. 
For $1 \le i \le 3$, let $V_i(p) = \{p' \in P \mid p' \in C_i(p), p' \ne p \}$ and $\overline{V_i}(p) = \{p' \in P \mid p' \in \overline{C_i}(p), p' \ne p \}$. For any two points $p$ and $q$, the smallest down-triangle containing $p$ and $q$ is denoted by $\bigtriangledown pq$ and the smallest up-triangle containing $p$ and $q$ is denoted by $\bigtriangleup pq$. If $G_1$ and $G_2$ are graphs on the same vertex set, $G_1 \cap G_2$ (resp. $G_1 \cup G_2$) denotes the graph on the same vertex set whose edge set is the intersection (resp. union) of the edge sets of $G_1$ and $G_2$.
\begin{figure}[h]
\centering
  \includegraphics[scale=0.6]{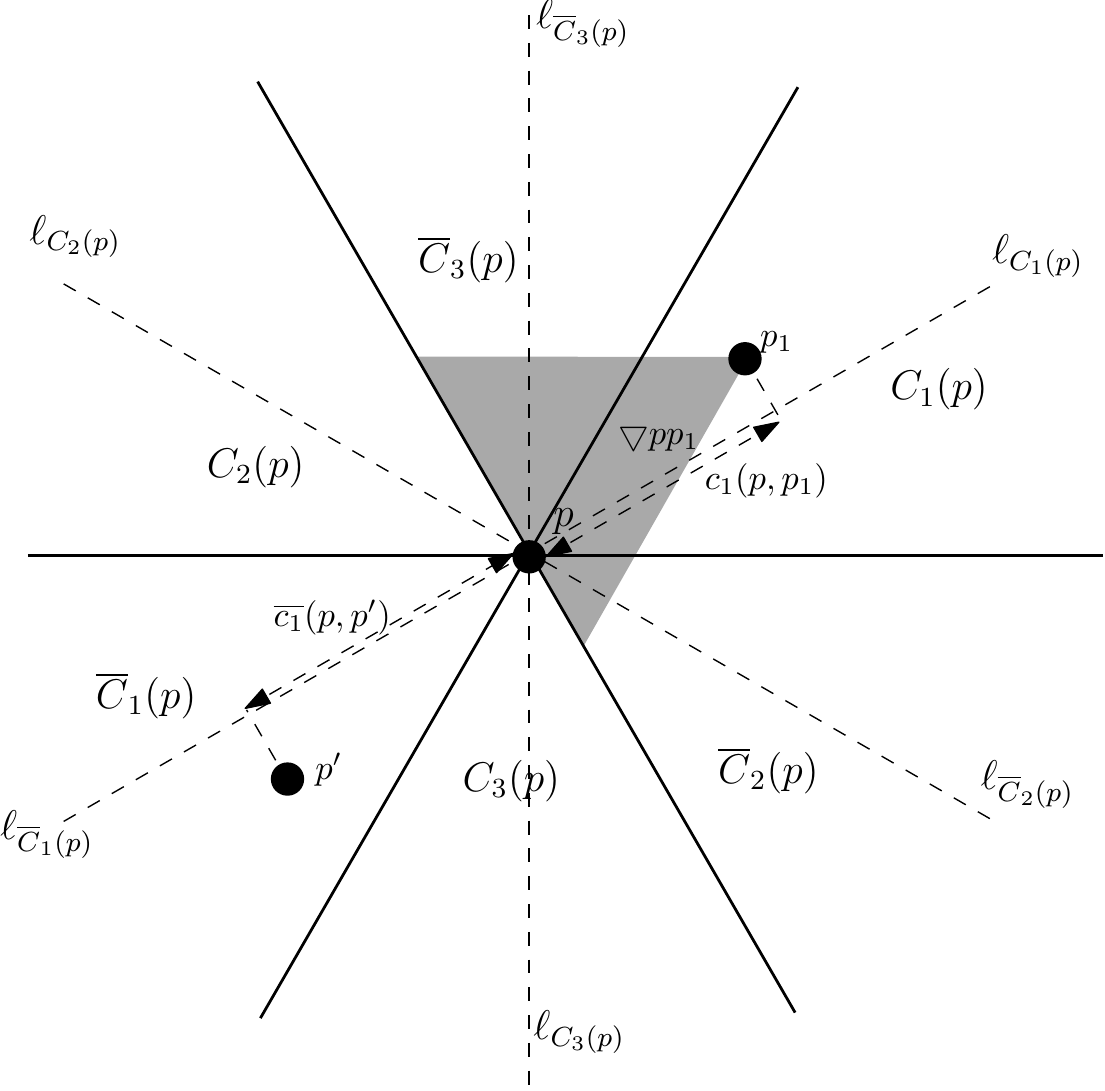}   
  \caption{Six angles around a point $p$.}
\label{Figcones}
  \end{figure}
\section{Preliminaries}\label{prelims}
In this section, we describe some basic properties of the geometric graphs described earlier and their equivalence with other geometric graphs which are well known in the literature. 

The class of down-triangles (and up-triangles) admits a shrinkability property \cite{Abrego2009}: each triangle object in this class that contains two points $p$ and $q$, can be shrunk such that $p$ and $q$ lie on its boundary. It is also clear that we can continue the shrinking process\textemdash from the edge that does not contain neither $p$ or $q$\textemdash until at least one of the points, $p$ or $q$, becomes a triangle vertex and the other point lies on the edge opposite to this vertex. After this, if we shrink the triangle further, it cannot contain $p$ and $q$ together. Therefore, for any pair of points $p$ and $q$, $\bigtriangledown pq$ ($\bigtriangleup pq$) has one of the points $p$ or $q$ at a vertex of $\bigtriangledown pq$ ($\bigtriangleup pq$) and the other point lies on the edge opposite to this vertex. In Fig. \ref{graph}, triangles are shown after shrinking. 

By the shrinkability property, for the $\bigtriangledown$-matching problem, it is enough to consider the smallest down-triangle for every pair of points $(p,q)$ from $P$. Thus, $G_{\bigtriangledown}(P)$ is equivalent to the graph whose vertex set is $P$ and that has an edge between any two vertices $p$ and $q$ if and only if $\bigtriangledown pq$ contains no other points from $P$. Notice that if $\bigtriangledown pq$ has $p$ as one of its vertices, then $q \in \overline{C_1}(p) \cup \overline{C_2}(p) \cup \overline{C_3}(p)$. The following two properties are simple, but useful.
\begin{property}\label{obs1}
 Let $p$ and $p'$ be two points in the plane. Let $i \in \{1, 2, 3\}$. The point $p$ is in the cone $C_i(p')$ if and only if the point $p'$ is in the cone $\overline{C}_i(p)$. Moreover, if $p$ is in the cone $C_i(p')$, then ${c_i}(p', p)=\overline{c_i}(p, p')$.
\end{property}
\begin{proof}
\begin{figure}
\centering
  \includegraphics[scale=0.6]{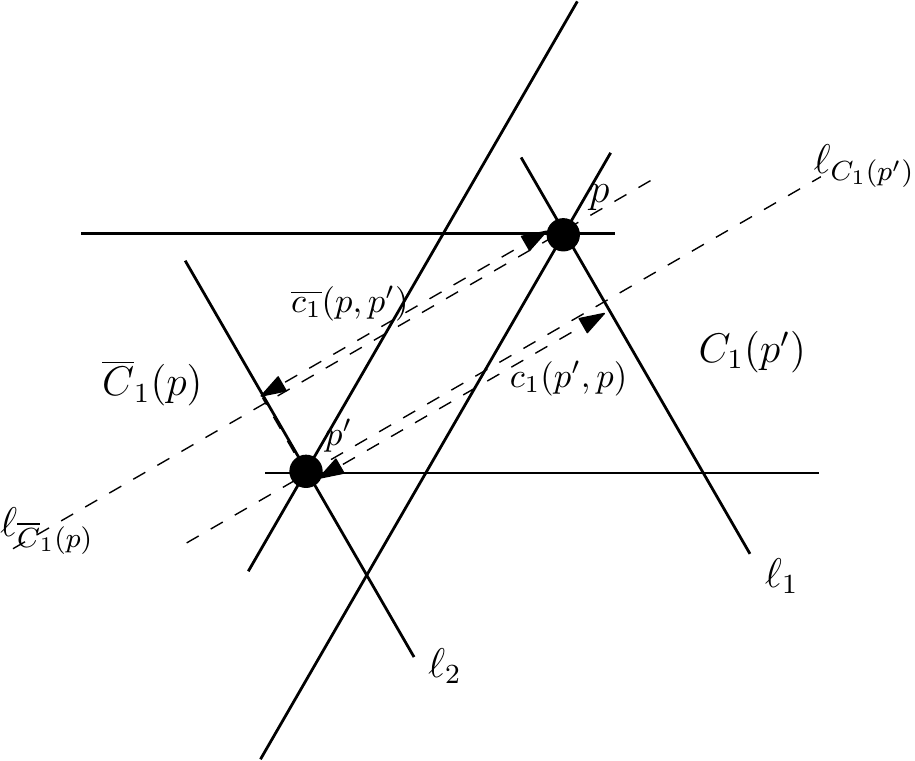}   
  \caption{Proof of Property \ref{obs1}.}
\label{Figcones2}
  \end{figure}
 The first part of the claim is obvious. Now, without loss of generality, assume that $i=1$ and $p \in C_1(p')$. (See Fig. \ref{Figcones2}.) Since $\ell_{\overline{C_1}(p)}$ is the bisector of $\overline{C_1}(p)$ and $\ell_{C_1(p')}$ is the bisector of $C_1(p')$, $\ell_{\overline{C_1}(p)}$ and $\ell_{C_1(p')}$ are parallel lines. Hence, $\overline{c_1}(p, p')$ is the perpendicular distance of $p'$ to the line $\ell_1$, which makes an angle $120^{\circ}$ with the horizontal and passes though $p$. Similarly, ${c_1}(p', p)$ is the perpendicular distance of $p$ to the line $\ell_2$, which makes an angle $120^{\circ}$ with the horizontal and passes though $p'$. Hence both $\overline{c_1}(p, p')$ and ${c_1}(p', p)$ are equal to the perpendicular distance between the lines 
$\ell_1$ and $\ell_2$.
\qed
\end{proof}
\begin{property}\label{obs2}
Let $P$ be a point set, $p \in P$ and $i \in \{1, 2, 3\}$. If $\overline{V}_i(p)$ is non-empty, then, 
in $G_\bigtriangledown(P)$, the vertex $p'$ corresponding to the point in $\overline{V}_i(p)$ 
with the minimum value of $\overline{c_i}(p, p')$ is the unique neighbour of vertex $p$ in $\overline{V}_i(p)$. 
\end{property}
\begin{proof}
Assume $\overline{V}_i(p) \ne \emptyset$. For any point $p'$ in $\overline{V}_i(p)$, it is easy to see that $\bigtriangledown pp'$ contains no points outside the cone $\overline{C_i}(p)$. Let $p'$ be the point with the minimum value of $\overline{c_i}(p, p')$. The minimality ensures that $\bigtriangledown pp'$ does not contain any other point other than $p$ and $p'$ from $P$. Therefore, $p$ and $p'$ are neighbours in $G_\bigtriangledown(P)$.

In order to prove uniqueness, consider any point $q$ in $P \cap \overline{V}_i(p)$ other than $p$ and $p'$. 
It can be seen that $\bigtriangledown pq$ contains the point $p'$ and therefore, $p$ and $q$ are not adjacent in $G_\bigtriangledown(P)$. 
Thus $p'$ is the only neighbour of $p$ in $\overline{V}_i(p)$.
\qed
\end{proof}
Consider a point set $P$ and let $p, q \in P$ be two distinct points. By Property \ref{obs1}, $\exists i \in \{1, 2, 3\}$ such that $p \in \overline{C_i}(q)$ or $q \in \overline{C_i}(p)$; by the general position assumption, both conditions cannot hold simultaneously. Since $\bigtriangledown pq$ has either $p$ or $q$ as a vertex, Property \ref{obs2} implies that we can construct $G_\bigtriangledown(P)$ as follows. For every point $p \in P$, and for each of the three cones, $\overline{C_i}$, for $i \in \{1, 2, 3\}$, add an edge from  $p$ to the point $p'$ in $\overline{V_i}(p)$ with the minimum value of $\overline{c_i}(p, p')$, if $\overline{V_i}(p)\ne \emptyset$. This definition of $G_\bigtriangledown(P)$ is the same as the definition of the half-$\Theta_6$-graph on negative cones ($\overline{C_i}$), given by Bonichon et al. \cite{Bonichon2010}. We can similarly define the graph $G_\bigtriangledown(P)$ using the cones ${C_i}$ instead of $\overline{C_i}$, for $i \in \{1, 2, 3\}$, and show that it is equivalent 
to the half-$\Theta_6$ graph on positive cones ($C_i$), given by Bonichon et al. \cite{Bonichon2010}. In Bonichon et al. \cite{Bonichon2010}, it was shown that for point sets in general position, the half-$\Theta_6$-graph, the {\em triangular distance-Delaunay graph} (TD-Del) \cite{Chew1989}, which are 2-spanners, and the {\em geodesic embedding} of $P$, are all equivalent. 

The $\Theta_k$-graphs discovered by Clarkson \cite{Clarkson1987} and Keil \cite{Keil1988} in the late 80's, are also used as spanners \cite{Narasimhan2007}. In these graphs, adjacency is defined as follows: the space around
each point $p$ is decomposed into $k \geqslant 2$ regular cones, each with apex $p$, and a point $q$ of a
given cone $C$ is linked to $p$ if, from $p$, the orthogonal projection of $q$ onto $C$'s bisector is the nearest point in $C$. In Bonichon et al. \cite{Bonichon2010}, it was shown that every $\Theta_6$-graph is the union of two half-$\Theta_6$-graphs, defined by $C_i$ and $\overline{C}_i$ cones. In our notation this is same as the graph $G_\bigtriangledown(P) \cup G_\bigtriangleup(P)$, which by definition, is equivalent to $G_{\davidsstar}(P)$. Thus, for a point set in general position, $\Theta_6(P) = G_{\davidsstar}(P)$. 

Now, we will prove some more properties of $G_\bigtriangledown(P)$ which will be used in the later sections of the paper.
\begin{property}\label{shortobs}
 Let $p \in P$ with $V_i(p) \ne \emptyset$, $\overline{V_j}(p) = \emptyset$, $\overline{V_k}(p) = \emptyset$ for distinct  
$i, j, k \in \{1, 2, 3 \}$. Then, in the graph $G_\bigtriangledown(P)$, $p$ has at least one neighbour in $V_i(p)$.
\end{property}
\begin{proof}
 Without loss of generality, assume that $i=1, j=2$ and $k=3$. It is easy to observe that, for any point $p_1 \in V_1(p)$, $\bigtriangledown pp_1 \subseteq C_1(p) \cup \overline{C_2}(p) \cup \overline{C_3}(p)$ (See Fig. \ref{Figcones}). Since $\overline{V_2}(p) = \emptyset$ and $\overline{V_3}(p) = \emptyset$, for any point $p_1 \in V_1(p)$, $P \cap \bigtriangledown pp_1 \subseteq V_1(p)\cup \{p\}$. 
To find a vertex in $V_1(p)$ which is a neighbour of $p$ in $G_\bigtriangledown(P)$, we just need to find a point $p_1 \in V_1(p)$ such that $\bigtriangledown pp_1$ contains no point from $V_1(p)$ other than $p_1$. 

We can choose any point $p_1 \in V_1(p)$ to start with. 
If $\bigtriangledown pp_1$ contains no point from $V_1(p)$ other than $p_1$, we are done.  
If not, replace $p_1$ with some other point inside $\bigtriangledown pp_1$ and repeat the process. Since triangle sizes are going down in each
step, eventually we will end up with a vertex in $V_1(p)$ such that $\bigtriangledown pp_1$ contains no point from $V_1(p)$ other than $p_1$. 
\qed
\end{proof}
\begin{property}\label{obs3}
 Let $p \in P$ with $V_i(p) \ne \emptyset$ and $\overline{V_i}(p) \ne \emptyset$ for some $i \in \{1, 2, 3 \}$. Then the
 vertex corresponding to $p$ has degree at least two in $G_\bigtriangledown(P)$.
\end{property}
\begin{proof}
Without loss of generality, assume that $V_1(p) \ne \emptyset$ and $\overline{V_1}(p) \ne \emptyset$. If $\overline{V_2}(p) \ne \emptyset$ or 
 $\overline{V_3}(p) \ne \emptyset$, then by Property \ref{obs2}, $p$ has a neighbour in $\overline{V_2}(p) \cup \overline{V_3}(p)$. 
On the other hand, if $\overline{V_2}(p) = \emptyset$ and $\overline{V_3}(p) = \emptyset$, then, by Property \ref{shortobs}, 
$p$ has at least one neighbour in $V_1(p)$. 

By Property \ref{obs2}, we know that $p$ has a unique neighbour in $\overline{V_1}(p)$ too. Thus, the degree of $p$ is at least two. 
\qed
\end{proof}
\begin{property}\label{obs4}
Let $p \in P$ be such that the vertex corresponding to $p$ is of degree one in $G_{\bigtriangledown}(p)$. 
Suppose $\exists i, j \in \{1, 2, 3\}, i \ne j$, such that $V_i(p) \ne \emptyset$ and $V_j(p) \ne \emptyset$. Let $k$ be the element in $\{1, 2, 3\} \setminus \{i, j\}$. Then, $\overline{V_k}(p) \ne \emptyset$.
\end{property}
\begin{proof}
Assume that $V_i(p) \ne \emptyset$ and $V_j(p) \ne \emptyset$, for distinct $i, j \in \{1, 2, 3\}$. 
If $\overline{V_i}(p) \ne \emptyset$ or $\overline{V_j}(p) \ne \emptyset$, then, by Property \ref{obs3}, the degree 
of $p$ is at least two in $G_\bigtriangledown(p)$, which is a contradiction. Therefore, $\overline{V_i}(p) \cup \overline{V_j}(p)=\emptyset$. 

If $\overline{V_k}(p) = \emptyset$, then by Property \ref{shortobs}, 
$p$ has at least one neighbour each in $V_i(p)$ and $V_j(p)$. If this is the case, the degree of $p$ is at least two, which is a contradiction. 
Therefore, $\overline{V_k}(p) \ne \emptyset$.
\qed
\end{proof}
\section{Some properties of $G_\bigtriangledown(P)$}
\subsection{Planarity}
Chew defined \cite{Chew1989} TD-Delaunay graph to be a planar graph and its equivalence with $G_\bigtriangledown(P)$ graph implies that $G_\bigtriangledown(P)$ is planar. This also follows from the general result that Delaunay graph of any convex distance function is a planar graph \cite{Bose}. For the sake of completeness, we include a direct proof here.   
\begin{lemma}\label{lmplanar}
 For a point set $P$, its $G_\bigtriangledown(P)$ is a plane graph, where its edges are straight line segments between the corresponding end-points.
\end{lemma}
\begin{proof}
Whenever there is an edge between $p$ and $q$ in $G_\bigtriangledown(P)$, we draw it as a straight line segment from $p$ to $q$. Notice that this segment always lies within $\bigtriangledown pq$. We will show that this gives a planar embedding of $G_\bigtriangledown(P)$. 
\begin{figure}[h]
\centering
 \includegraphics[scale=0.6]{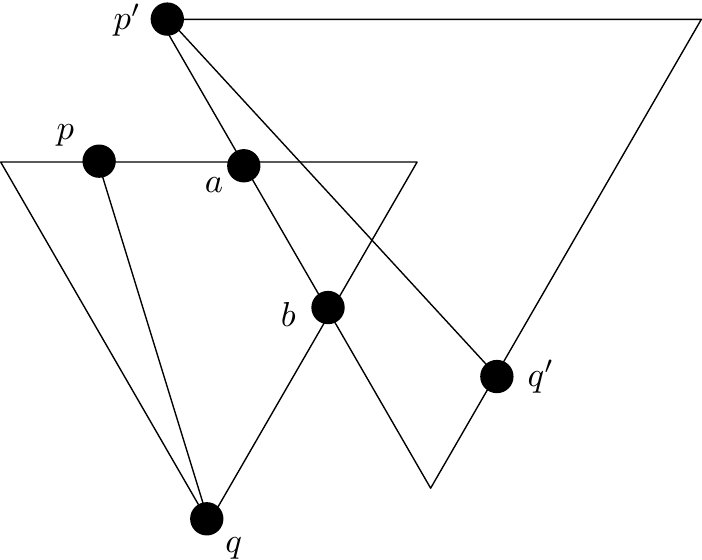}   
  \caption{Intersection of $\bigtriangledown pq$ and $\bigtriangledown p'q'$ does not lead to crossing of edges $pq$ and $p'q'$.}
\label{planarity}
 \end{figure}
Consider two edges $pq$ and $p'q'$ of $G_\bigtriangledown(P)$. If the interiors of $\bigtriangledown pq$ and $\bigtriangledown p'q'$ have no point in common, the line segments $pq$ and $p'q'$ can not cross each other. Suppose the interiors of $\bigtriangledown pq$ and $\bigtriangledown p'q'$ share some common area. The case that $\bigtriangledown pq \subseteq \bigtriangledown p'q'$ (or vice versa) is not possible, because in this case $\bigtriangledown p'q'$ contains $p$ and $q$ (or $\bigtriangledown pq$ contains $p'$ and $q'$), which contradicts its emptiness. Since $\bigtriangledown pq$ and $\bigtriangledown p'q'$ have parallel sides, this implies that one corner of $\bigtriangledown pq$ infiltrates into $\bigtriangledown p'q'$ or vice versa (see Fig. \ref{planarity}). Thus their boundaries cross at two distinct points, $a$ and $b$. Since $P \cap \bigtriangledown p'q' \cap \bigtriangledown p'q' =\emptyset$, the points $p$ and $q$ must be on that portion of the boundary of $\bigtriangledown pq$ that does not 
lie inside $\bigtriangledown p'q'$. So the 
line through $ab$ separates $pq$ from $p'q'$. 
\qed
\end{proof}
Throughout this paper, we use $G_\bigtriangledown(P)$ to represent both the abstract graph and its planar embedding described in Lemma \ref{lmplanar}. The meaning will be clear from the context.
 \subsection{Connectivity}
In this section, we prove that for a point set $P$, its $G_\bigtriangledown(P)$ is connected. As stated in the following lemma, between every pair of vertices, there exist a path with a special structure.  
\begin{lemma}\label{pathintriangle}
 Let $P$ be a point set with $p, q \in P$. Then, in $G_\bigtriangledown(P)$, there 
is a path between $p$ and $q$ which lies fully in $\bigtriangledown pq$ and hence $G_\bigtriangledown(P)$ is connected. 
\end{lemma}
\begin{proof}
 We will prove this using induction on the area of $\bigtriangledown pq$. For any pair of distinct points $p, q \in P$, if the interior of $\bigtriangledown pq$ does not contain any point from $P$, by definition, there is an edge from $p$ to $q$ in $G_\bigtriangledown(P)$. By induction, assume that for  pairs of points $x, y \in P$ such that the area of $\bigtriangledown xy$ is less than the area of $\bigtriangledown pq$, in the graph in $G_\bigtriangledown(P)$, there is a path  which lies fully in $\bigtriangledown xy$ between $x$ and $y$.

If the interior of $\bigtriangledown pq$ does not contain any point from $P$, there is an edge from $p$ to $q$ in $G_\bigtriangledown(P)$. Otherwise, there is a point $x \in P$ which is in the interior of $\bigtriangledown pq$. This implies $\bigtriangledown px \subset \bigtriangledown pq$ and $\bigtriangledown xq \subset \bigtriangledown pq$. Since the area of $\bigtriangledown px$ and the area of $\bigtriangledown xq$ are both less than the area of $\bigtriangledown pq$, by the induction hypothesis, there is a path that lies in $\bigtriangledown px$ between $p$ and $x$ and there is a path that lies in $\bigtriangledown xq$ between $x$ and $q$. By concatenating these two paths, we get a path which lies in $\bigtriangledown pq$ between $p$ and $q$.  
\qed
\end{proof}
\subsection{Number of degree-one vertices}
In this section, we prove for a point set $P$, its $G_\bigtriangledown(P)$ has at most three vertices 
of degree one. This fact is important for our proof of the lower bound of the cardinality of a maximum matching in $G_\bigtriangledown(P)$.
\begin{lemma} \label{degreeone}
For a point set $P$, its $G_\bigtriangledown(P)$ has at most three vertices 
of degree one.
\end{lemma}
\begin{proof}
We will give a proof by contradiction. Let $p_1, p_2, p_3$ and $p_4$ be four points such that the vertices corresponding to them are of degree one in $G_\bigtriangledown(p)$. Since the points are in general position, without loss of generality, we can assume that these points are given in the bottom to top order of their $y$ co-ordinates. We analyse different relative positionings of $p_2$ and $p_3$ with respect to $p_1$ and prove that in none of these cases, we can properly place all the four points consistently. Since $p_1$ is below $p_2$ and $p_3$, the relative positioning of $p_2$ and $p_3$ should be one of the following :
\begin{itemize}
 \item Case 1 : $p_2 \in \overline{V_3}(p_1)$.
 \item Case 2 : $p_2 \notin \overline{V_3}(p_1)$ but $p_3 \in \overline{V_3}(p_1)$.
 \item Case 3 : $p_2, p_3 \in V_1(p_1)$ or $p_2, p_3 \in V_2(p_1)$.
 \item Case 4 : $p_2 \in V_1(p_1)$, $p_3 \in V_2(p_1)$ or $p_2 \in V_2(p_1)$, $p_3 \in V_1(p_1)$.
\end{itemize}
\textbf{Case 1.} Since $p_2 \in \overline{V_3}(p_1)$, we have $p_1 \in V_3(p_2)$. 
Since $p_2$ is of degree one, by Property \ref{obs3}, $\overline{V_3}(p_2) = \emptyset$. Since $p_3$ and $p_4$ are above $p_2$, and $p_4$ is above $p_3$, we have only the following sub-cases 
to consider. (See Fig. \ref{figs1234}.)
\begin{itemize}
 \item Case 1a. $p_3, p_4 \in V_i(p_2)$ and $p_4 \in V_i(p_3)$, where $i \in \{1, 2\}$. 
 \item Case 1b. $p_3, p_4 \in V_i(p_2)$, where $i \in \{1, 2\}$, and $p_4 \in \overline{V_3}(p_3)$. 
 \item Case 1c. $p_3, p_4 \in V_i(p_2)$ and $p_4 \in V_j(p_3)$, where $i,j \in \{1, 2\}$ and $i \ne j$. 
  \item Case 1d. $p_3 \in V_i(p_2)$, $p_4 \in V_j(p_2)$, where $i, j \in \{1,2\}, i \ne j$. 
\end{itemize}
\begin{figure}[h]
  \centering
  \includegraphics[scale=0.6]{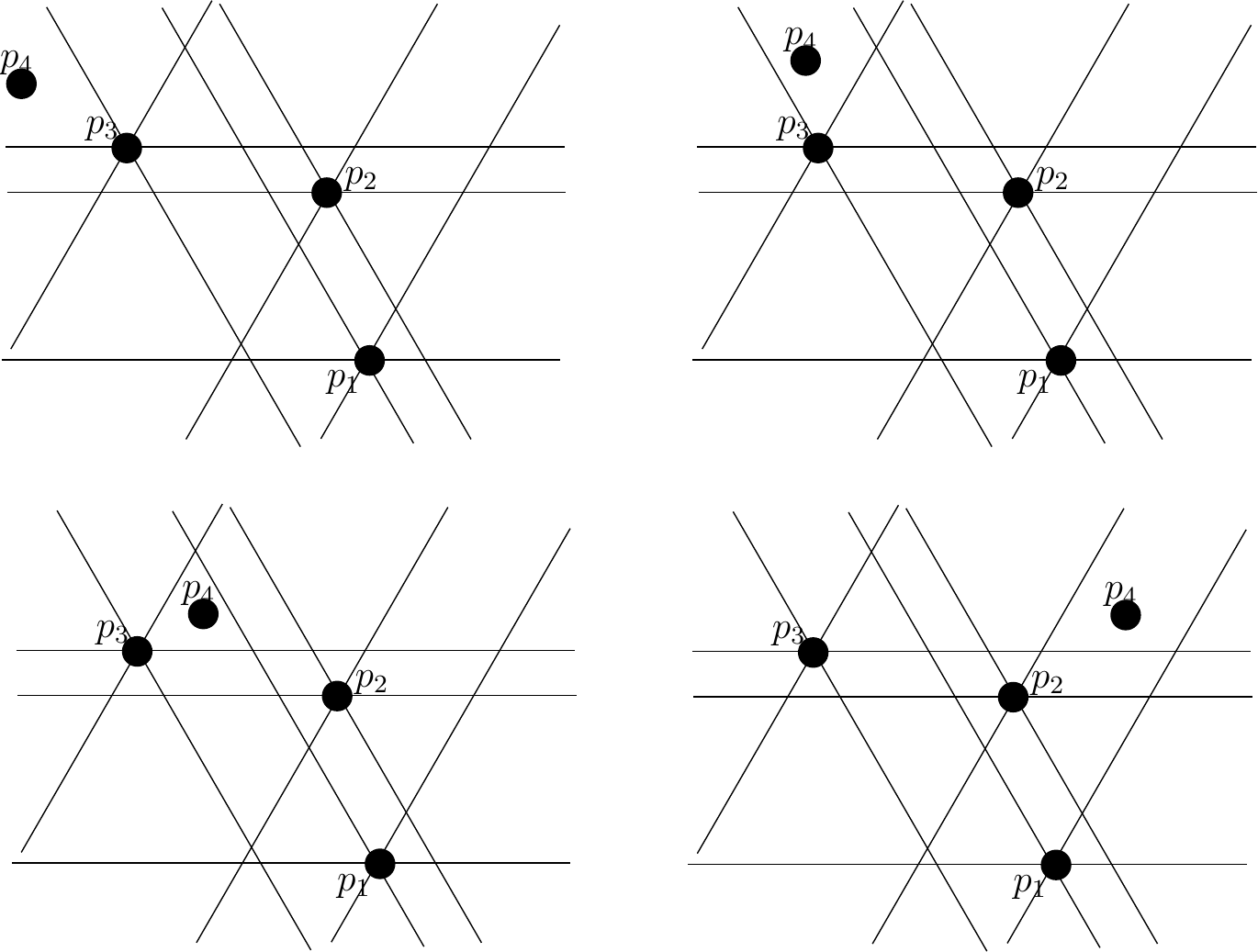}   
  \caption{Sub-cases of Case 1}
\label{figs1234}
  \end{figure}
Without loss of generality, assume that $i=2$ and $j=1$. \newline \newline
Case 1a : We have $p_3, p_4 \in V_2(p_2)$, implying that $p_2 \in \overline{V_2}(p_3)$ and $p_2 \in \overline{V_2}(p_4)$. 
Since $p_4 \in V_2(p_3)$ and $p_2 \in \overline{V_2}(p_3)$, by Property \ref{obs3}, the degree of $p_3$ is at least two.
This is a contradiction. \newline
Case 1b : We have $p_3, p_4 \in V_2(p_2)$. This implies that $p_2 \in \overline{V_2}(p_3)$ and $p_2 \in \overline{V_2}(p_4)$. 
Since $p_4 \in \overline{V_3}(p_3)$ and $p_2 \in \overline{V_2}(p_3)$, by Property \ref{obs2}, the degree of $p_3$ is at least two. This 
is a contradiction.\newline
Case 1c : We have $p_3, p_4 \in V_2(p_2)$ . This implies that $p_2 \in \overline{V_2}(p_3)$ and $p_2 \in \overline{V_2}(p_4)$.
Since $p_4 \in V_1(p_3)$, we have $p_3 \in \overline{V_1}(p_4)$. Since we already had $p_2 \in \overline{V_2}(p_4)$, 
by Property \ref{obs2}, the degree of $p_4$ is at least two, which is a contradiction. \newline
Case 1d : Since $p_3 \in V_2(p_2)$ and $p_4 \in V_1(p_2)$, by Property \ref{obs4}, we get $\overline{V_3}(p_2) \ne \emptyset$. This contradicts the property $\overline{V_3}(p_2) = \emptyset$, that we made at the beginning of the analysis of Case 1.\newline
\textbf{Case 2.} Without loss of generality, assume that $p_2 \in V_2(p_1)$. Thus, $p_1 \in \overline{V_2}(p_2)$. Since $\overline{V_2}(p_2) \ne \emptyset$, if we have $V_2(p_2) \ne \emptyset$, by Property \ref{obs3}, the degree of $p_2$ is at least two. On the other hand, if $\overline{V_3}(p_2) \ne \emptyset$, by Property \ref{obs2}, the degree of $p_2$ is at least two. Since both cases lead to contradictions, we have $V_2(p_2) = \emptyset$ and $\overline{V_3}(p_2) = \emptyset$. Since $p_3$ and $p_4$ are above $p_2$, this implies that $p_3, p_4 \in V_1(p_2)$. This gives us $p_2 \in \overline{V_1}(p_4)$ and $p_2 \in \overline{V_1}(p_3)$ (See Fig. \ref{figs567}.).
Since $p_2 \in \overline{V_1}(p_3)$ and $p_3$ is of degree one, by Property \ref{obs3}, we get $V_1(p_3) = \emptyset$ and by Property \ref{obs2}, we get $\overline{V_3}(p_3) = \emptyset$. Since $p_4$ is above $p_3$, this implies that $p_4 \in V_2(p_3)$ and hence $p_3 \in \overline{V_2}(p_4)$. Since we already had $p_2 \in \overline{V_1}(p_4)$, by Property \ref{obs2}, the degree of $p_4$ is at least two, which is a contradiction.\newline
\begin{figure}[h]
  \centering
  \includegraphics[scale=0.6]{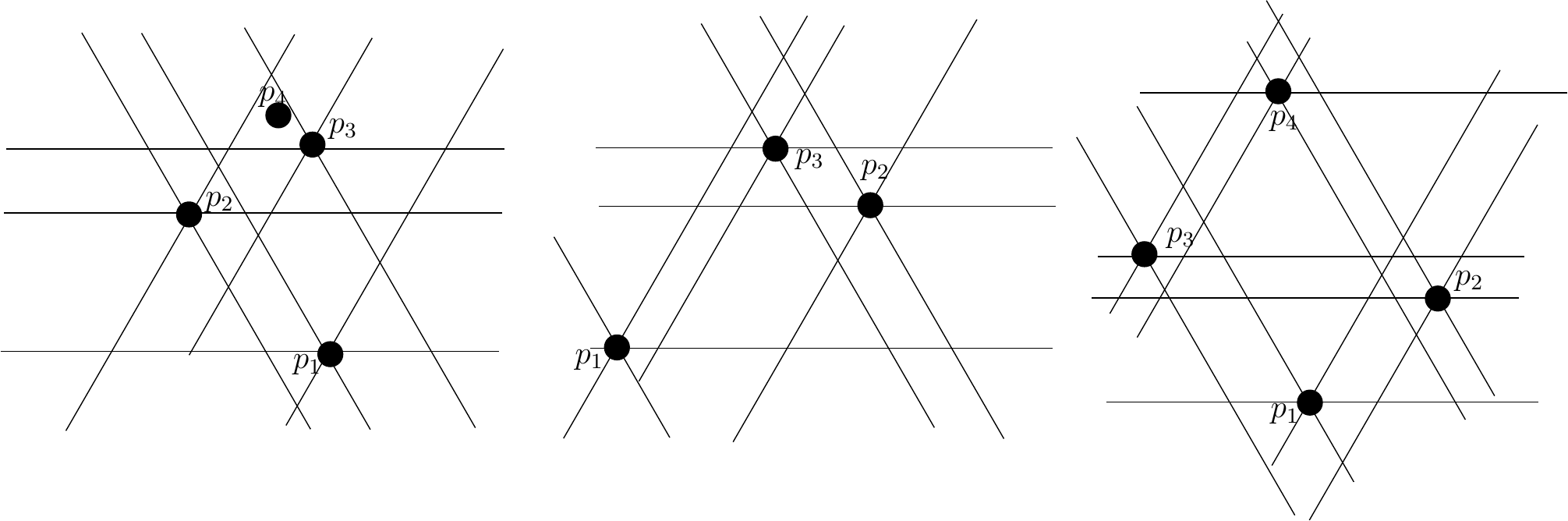}   
  \caption{Case 2, Case 3 and Case 4.}
\label{figs567}
  \end{figure}
\textbf{Case 3.} We need only consider the situation $p_2, p_3 \in V_1(p_1)$. The other situation is symmetric to this. Since $p_2 \in V_1(p_1)$, we get $p_1 \in \overline{V_1}(p_2)$. Since $\overline{V_1}(p_2) \ne \emptyset$, and $p_2$ is of degree one, by Properties \ref{obs2} and \ref{obs3}, we get $V_1(p_2) = \emptyset$ and $\overline{V_3}(p_2) =\emptyset$. 

Since $p_3$ is above $p_2$, this means that $p_3 \in V_2(p_2)$, which gives $p_2 \in \overline{V_2}(p_3)$. Since $p_3 \in V_1(p_1)$ by assumption, we also have $p_1 \in \overline{V_1}(p_3)$ (See Fig. \ref{figs567}.). These two observations give us $\overline{V_1}(p_3) \ne \emptyset$ and $\overline{V_2}(p_3) \ne \emptyset$. Applying Property \ref{obs2}, it follows that the degree of $p_3$ is at least two, which is a contradiction.\newline
\textbf{Case 4.} We need only consider the situation $p_2 \in V_1(p_1)$ and $p_3 \in V_2(p_1)$. The other situation is symmetric. Since $p_2 \in V_1(p_1)$, we have $p_1 \in \overline{V_1}(p_2)$. Since $p_2$ is of degree one, and $\overline{V_1}(p_2) \ne \emptyset$, by Properties \ref{obs2} and \ref{obs3}, we get $V_1(p_2) = \emptyset$ and $\overline{V_3}(p_2) =\emptyset$. 
Since $p_4$ is above $p_2$, this means that $p_4 \in V_2(p_2)$, which gives $p_2 \in \overline{V_2}(p_4)$. 

Similarly, since $p_3 \in V_2(p_1)$, we have $p_1 \in \overline{V_2}(p_3)$. Since $p_3$ is of degree one, and $\overline{V_2}(p_3) \ne \emptyset$, by Properties \ref{obs2} and \ref{obs3}, we get $V_2(p_3) = \emptyset$ and $\overline{V_3}(p_3) =\emptyset$. Since $p_4$ is above $p_3$, this means that $p_4 \in V_1(p_3)$, which gives $p_3 \in \overline{V_1}(p_4)$. Since we already had $p_2 \in \overline{V_2}(p_4)$, using Property \ref{obs2}, it follows that the degree of $p_4$ is at least two, which is a contradiction (See Fig. \ref{figs567}.).

Thus in each of the four possible placements of $p_1, p_2, p_3$ and $p_4$, we concluded that the configuration is impossible. This completes the proof.
\qed
\end{proof}
\subsection{Internal triangulation}
In this section, we will prove that for a point set $P$, the plane graph $G_\bigtriangledown(P)$ is internally triangulated. This property will be  used in Section \ref{maxmatch} to derive the lower bound for the cardinality of maximum matchings in $G_\bigtriangledown(P)$.
\begin{lemma}\label{internal}
 For a point set $P$, all the internal faces of $G_\bigtriangledown(P)$ are triangles.
\end{lemma}
\begin{proof}
Consider an internal face $f$ of $G_\bigtriangledown(P)$. We need to show that $f$ is a triangle. Let $p$ be the vertex with the highest $y$-coordinate among the vertices on the boundary of $f$. Since $f$ is an internal face, $p$ has at least two neighbours on the boundary of $f$. Let $q$ and $r$ be the neighbours of $p$ on the boundary of $f$ such that $r$ is to the right of the line passing through $q$ and making an angle of $120^{\circ}$ with the horizontal and any other neighbour of $p$ on the boundary of $f$ is to the right of the line passing through $r$ and making an angle $120^{\circ}$ with the horizontal. Because of the general position assumption, $q$ and $r$ can be uniquely determined. 

We will prove that $qr$ is also an edge on the boundary of $f$ and there is no point from $P$ in the interior of the triangle whose vertices are $p, q$ and $r$. This will imply that the face $f$ is the triangle whose vertices are $p, q$ and $r$. 

We know that $q, r \in \overline{C_1}(p) \cup \overline{C_2}(p) \cup C_3(p)$. By Property \ref{obs2},
it cannot happen that both $q, r \in \overline{C_i}(p)$, for any $ i \in \{1, 2\}$.
Other possibilities are shown in Fig. \ref{figtriangles}, where $q$ is assumed to be above $r$.
An analogous argument can be made when $r$ is above $q$ as well.
\begin{figure}
\centering
  \includegraphics[scale=0.58]{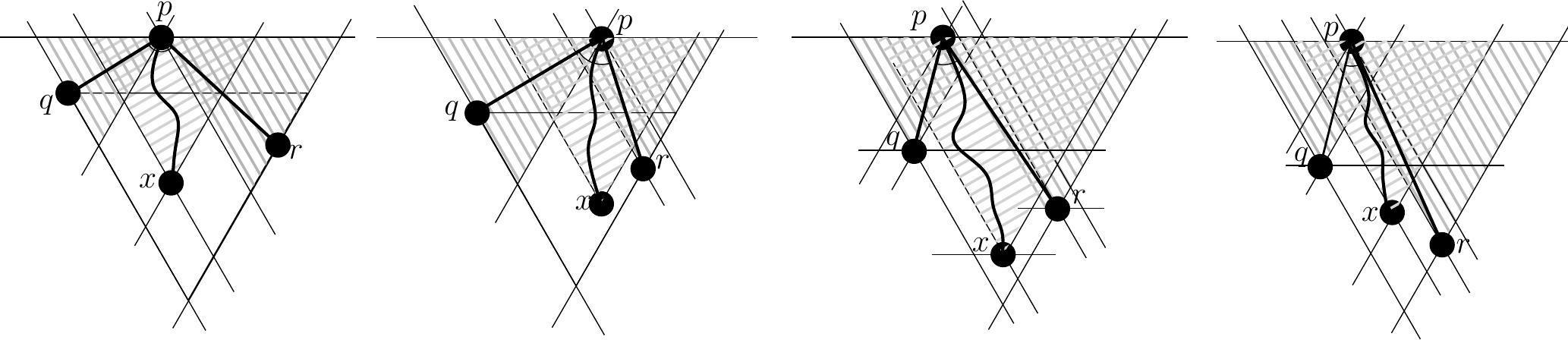}   
  \caption{Case 1. $q \in \overline{C_1}(p)$ and $r \in \overline{C_2}(p)$, Case 2. $q \in \overline{C_1}(p)$ and $r \in C_3(p)$, Case 3. $r \in \overline{C_2}(p)$ and $q \in C_3(p)$, Case 4. $q, r \in C_3(p)$.}
\label{figtriangles}
  \end{figure}
Since $pq$ and $pr$ are edges in $G_\bigtriangledown(P)$, we know that $\bigtriangledown pq \cap (P \setminus \{p, q\}) =\emptyset$ and
$\bigtriangledown pr \cap (P \setminus \{p, r\})=\emptyset$. 

Notice that, the area bounded by the lines (1) the horizontal line passing through $p$, (2) the line passing through $q$ and making an angle of $120 ^{\circ}$ with the horizontal, and (3) the line passing through $r$ and making an angle of $60 ^{\circ}$ with the horizontal, will define an equilateral down triangle with $p$, $q$ and $r$ on its boundary. Let us denote this triangle by $\bigtriangledown pqr$.
\begin{claim}
$\bigtriangledown pqr \cap (P\setminus \{p, q, r\}) =\emptyset$ .
\end{claim}
\begin{proof}
For contradiction, let us assume that there exists a point $x \in \bigtriangledown pqr \cap (P\setminus \{p, q, r\})$. Because of the general position assumption, $x$ cannot be on the boundary of $\bigtriangledown pqr$. Therefore, $\bigtriangledown px$ does not contain $q$ and $r$. By Lemma \ref{pathintriangle}, in $G_\bigtriangledown(P)$,
there exists a path between $p$ and $x$ which lies inside $\bigtriangledown px$. Let this path be $X=v_1 v_2, \ldots, v_k=x$. Since $\bigtriangledown pq \cap P \setminus \{p, q\} =\emptyset$, $\bigtriangledown pr \cap P \setminus \{p, r\}=\emptyset$ and $q, r \notin \bigtriangledown px$, we know that all vertices in the path $X=v_1 v_2, \ldots, v_k=x$ lie inside
the region $R = (\bigtriangledown px \setminus (\bigtriangledown pq \cup \bigtriangledown pr)) \cup \{p\}$. 

Let $C$ be the cone with apex $p$ bounded by the rays $pq$ and $pr$. Observe that for any point $v \in R$, the line segment
$pv$ lies inside the cone $C$. Since $v_2 \in R$ and $pv_2$ is an edge (in the path from $p$ to $x$), the line segment corresponding to the edge $pv_2$ lies inside $C$ in $G_\bigtriangledown(P)$. 

If the point $v_2$ is outside the face $f$, edge $pv_2$ will cross the boundary of $f$, which is contradicting
the planarity of $G_\bigtriangledown(P)$. Since $v_2$ cannot be outside the face $f$, the edge $pv_2$ belongs to the boundary of $f$. Since $v_2$ lies inside the cone $C$ and $v_2\in R$, this means that $v_2$ is a neighbour of $p$ on the boundary of $f$ such that $v_2$ is to the left of the the line passing through $r$ and making an angle of $120 ^{\circ}$ with the horizontal. This is a contradiction to our assumption that $q$ is the only neighbour of $p$ on the boundary of $f$, lying to the left of the the line passing through $r$ and making an angle of $120 ^{\circ}$ with the horizontal. 
\qed
\end{proof}
Let us continue with the proof of Lemma \ref{internal}. Since the triangle with vertices $p, q$ and $r$ is inside the triangle $\bigtriangledown pqr$, from the above claim, it is clear that
there is no point from $P$, other than the points $p, q$ and $r$, inside the triangle whose vertices are $p, q$ and $r$. Since the edges $pq$ and $pr$ belong to the boundary of $f$, to show that $f$ is a triangle, it is now enough to prove that $qr$ is also an edge in $G_\bigtriangledown(P)$. This fact also follows from the above claim as explained below.

Since $\bigtriangledown qr \subseteq \bigtriangledown pqr$, by the claim above, $\bigtriangledown qr$ cannot contain any point from $P$ other than $p, q$ and $r$. Moreover, since $p$ lies above $q$ and $r$, we know that $p \notin \bigtriangledown qr$. Therefore, $\bigtriangledown qr \cap (P \setminus \{q, r\}) = \emptyset$. Therefore, $qr$ is an edge in $G_\bigtriangledown(P)$.

Thus, $f$ has to be a triangle bounded by the edges $pq$, $qr$ and $pr$. 
\qed
\end{proof}
\begin{corollary}\label{outercutvertex}
 For a point set $P$, all the cut vertices of $G_\bigtriangledown(P)$ lie on its outer face. 
\end{corollary}
\begin{proof}
 Consider any vertex $v$ of $G_\bigtriangledown(P)$ which is not on its outer face. Since $G_\bigtriangledown(P)$ is internally triangulated, each neighbour of $v$ in $G_\bigtriangledown(P)$ lies on a cycle in the graph $G_\bigtriangledown(P) \setminus v$. Since $G_\bigtriangledown(P)$
is connected, $G_\bigtriangledown(P) \setminus v$ remains connected. Thus, $v$ cannot be a cut vertex. 
\qed
\end{proof}
\section{Maximum matching in $G_\bigtriangledown(P)$}\label{maxmatch}
In this section, we show that for any point set $P$ of $n$ points, $G_\bigtriangledown(P)$ contains a matching of size $\lceil\frac{n-2}{3} \rceil$; i.e, at least $2(\lceil\frac{n-2}{3} \rceil)$ vertices are matched. Consider a point set $P$ containing $n$ points.  If we have only two points in $P$, then the graph contains a perfect matching. Hence, we assume that $|P| \ge 3$. 

We construct a graph $G'$ such that it is a $2$-connected planar graph of minimum degree at least $3$ and then make use of the following theorem of Nishizeki \cite{Nishi} to get a lower bound on the size of a maximum matching of $G'$. Using this, we will then derive a lower bound on the size of a maximum matching of $G_\bigtriangledown(P)$.
\begin{theorem}[\cite{Nishi}]\label{thmnishi}
 Let $G$ be a connected planar graph with $n$ vertices having minimum degree at least $3$ and let $M$ be a maximum matching in $G$. Then,
$$
|M| \ge \left\{ \begin{array}{rl}
 \lceil \frac{n+2}{3} \rceil &\mbox{when $n \ge 10$ and $G$ is not 2-connected} \\
  \lceil \frac{n+4}{3}\rceil &\mbox{when $n \ge 14$ and $G$ is 2-connected} \\
  \lfloor \frac{n}{2}\rfloor &\mbox{otherwise}  
       \end{array} \right.
$$
\end{theorem}
Initialize $G'$ to be the same as $G_\bigtriangledown(P)$. Consider a simple closed curve $\mathcal{C}$ in the plane such that (1) the entire graph $G_\bigtriangledown(P)$ (all vertices and edges) lies inside the bounded region enclosed by $\mathcal{C}$, (2)
the vertices of $G_\bigtriangledown(P)$ which lie on $\mathcal{C}$ are precisely the degree-one vertices of $G_\bigtriangledown(P)$,
(3) except for the end points, every edge of $G_\bigtriangledown(P)$ lies in the interior of the bounded region enclosed by $\mathcal{C}$.

Let the degree-one vertices of $G_\bigtriangledown(P)$ be denoted by $p_0, p_1, \ldots, p_{k-1}$. In the previous section, we proved that $k \le 3$. 

If $k \ge 2$, the region of the outer face of $G_\bigtriangledown(P)$ bounded by the curve $\mathcal{C}$ can be divided into $k$ regions $R_0,\ldots, R_{k-1}$ where $R_i$ is the region bounded by the edge at $p_{i}$, the edge at $p_{(i+1)\mod k}$, the boundary of the outer face of $G_\bigtriangledown(P)$ and the curve $\mathcal{C}$. See Fig. \ref{fignewgraph}. (Here onwards, in this subsection we assume that indices of vertices and regions are taken modulo $k$.) Notice that every vertex on the outer-face of $G_\bigtriangledown(P)$ lies on at least one of these regions and $p_i$ lies on the regions $R_i$ and $R_{i-1}$, for $0 \le i \le k-1$. We insert $k$ new vertices $x_0, \ldots, x_{k-1}$ into $G'$. (To visualize the abstract graph $G'$, vertex $x_i$ may be assumed to lie on the boundary of the region $R_i$, a point distinct from $p_{i}$ and $p_{i+1}$.) New edges are added between $x_i$ and $x_{i+1}$, for $0\le i\le k-1$. We also insert new edges into $G'$ between each $x_i$ and all the vertices of $G_\bigtriangledown(P)$ which lie on the region $R_i$, for $0 \le i \le k-1$. This 
transformation maintains planarity. (Edges between new vertices and old vertices can be drawn inside the corresponding region $R_i$. The edges among the new vertices can be drawn outside these regions, except at their end points.)

\begin{figure}[h]
  \centering
  \includegraphics[scale=0.6]{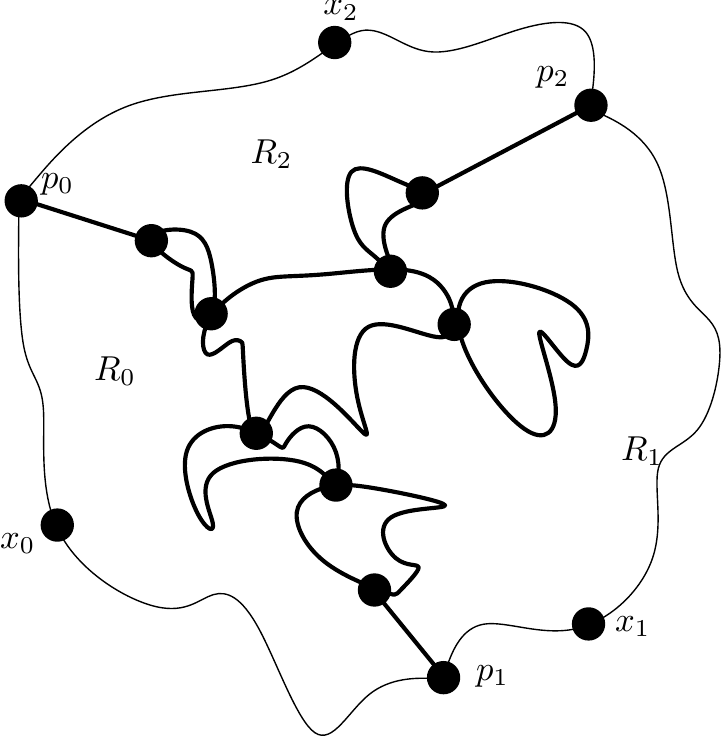}   
  \caption{Regions on the outer face of $G_\bigtriangledown(P)$.}
\label{fignewgraph}
  \end{figure} 

Each degree-one vertex $p_i$, $0 \le i \le k-1$, of $G_\bigtriangledown(P)$ lies on two regions $R_i$ and $R_{i-1}$, in $G'$ it gets two new edges; one to $x_i$ and the other to $x_{i-1}$. Thus the degree of $p_i$ becomes $3$. All other vertices on the outer face of $G_\bigtriangledown(P)$ were of degree at least two. Since they belong to at least one of the regions $R_0, \ldots, R_{k-1}$, they get at least one new edge in $G'$ and their degree is at least three in $G'$. Since $G_\bigtriangledown(P)$ is an internally triangulated planar graph, we know that all vertices except those on the outer face had degree at least $3$. These vertices maintain the same degrees in $G'$ as in $G$. The degree of $x_i$, $0 \le i \le k-1$, is also at least $3$ in $G'$, because it is adjacent to $p_i, p_{i+1}$ and at least one more vertex on the outer face of $G_\bigtriangledown(P)$. Thus, $G'$ has minimum degree at least three. 

If $k = 0$ or $1$, the modification of $G'$ is similar but rather simpler. We insert a new vertex $x$ in the outer face of $G'$ and add edges between $x$ and all other vertices in the outer face of $G_\bigtriangledown(P)$. This transformation maintains planarity. As earlier, all vertices in $G'$ except the vertex $p_0$ (present only when $k=1$) have degree at least three now. 

If $k=1$, the degree of $p_0$ has become two in $G'$ at this stage. In this case, let $f$ be a face of the current graph $G'$ which contains $p_0$ and $x$. Modify $G'$ by inserting a new vertex $y$ inside $f$ and adding edges from this new vertex to all other vertices belonging to $f$. As earlier, this transformation maintains planarity. Now, the degree of $p_0$ becomes $3$.

\begin{claim}
The graph $G'$ is $2$-connected. 
\end{claim}
\begin{proof}
It is easy to observe that none of the newly inserted vertices can be a cut vertex of $G'$. For any vertex $v$ which was not on the outer face of $G_\bigtriangledown(P)$, the induced subgraph on its neighbours form a cycle in $G'$ as it was in $G_\bigtriangledown(P)$. They cannot be cut vertices. 

Consider any vertex $v$ which was on the outer face of $G_\bigtriangledown(P)$. Suppose $G' \setminus v$ is not connected and let $C_1$ and $C_2$ be two connected components of $G' \setminus v$, with vertex sets $V_1$ and $V_2$ respectively. Let $G_i$ be the induced subgraph of $G_\bigtriangledown(P)$ on vertex set $V_i\cup \{v\}$, for $i \in \{1, 2\}$. We know that $G_i$ is connected and there exists at least one vertex $v_i$ other than $v$ which lies on the outer face of $G_\bigtriangledown(P)$. In $G'$, the vertex $v_i$ has an edge to at least one of the newly inserted vertices. Since the induced subgraph of $G'$ on the newly inserted vertices is connected, in $G'$ we get a path from $v_1$ to $v_2$ in which all the intermediate vertices are newly inserted vertices in $G'$. This means there is path from $C_1$ to $C_2$ in $G' \setminus v$, which is a contradiction. Therefore, no vertex on the outer face of $G_\bigtriangledown(P)$ can be a cut vertex in $G'$ and thus, $G'$ is $2$-connected.
\qed
\end{proof}
Thus, the graph $G'$ is a $2$-connected planar graph of minimum degree at least $3$, having at most $n+3$ vertices. Let $n'=n+k$ be the number of vertices of $G'$. By Theorem \ref{thmnishi}, the cardinality of a maximum matching $M'$ in $G'$ is at least $\left\lceil \frac{n'+4}{3}\right \rceil$ when $n' \ge 14$ and $|M'| \ge \lfloor \frac{n}{2} \rfloor$, otherwise. Since $G_\bigtriangledown(P)$ is an induced subgraph of $G'$, if we delete the edges in $M'$ which have at least one end point which is not in $P$, we get a matching $M$ of $G_\bigtriangledown(P)$. We have $|M| \ge |M'| - k$, where $k = n' - n$ with $k \in \{0, 1, 2, 3 \}$.  
From this, we get, 
$$
|M| \ge \left\{ \begin{array}{rl}
 \left\lceil \frac{n+4-2k}{3}\right \rceil &\mbox{when $n \ge 14 -k$} \\
 \left\lfloor \frac{n-k}{2} \right\rfloor &\mbox{otherwise}  
       \end{array} \right.
$$
Whenever $n \ge 7$, from the above inequality, we get $|M| \ge \left\lceil\frac{n-2}{3}\right\rceil \ge 2$. When $n \ge 5$, Lemma \ref{degreeone}
implies that $G_\bigtriangledown(P)$ cannot be a star with $n-1$ leaves and hence $|M| \ge 2$. When $n >1$, since  $G_\bigtriangledown(P)$ is connected, we get $|M| \ge 1$. From this discussion, we can conclude that, in all cases, $|M| \ge \left\lceil\frac{n-2}{3}\right\rceil$.
\begin{theorem}\label{matchingbound}
 For any point set $P$ of $n$ points in general position, $G_\bigtriangledown(P)$ contains a matching of size $\lceil\frac{n-2}{3} \rceil$.
\end{theorem}
\begin{figure}[h]
  \centering
  \includegraphics[scale=0.6]{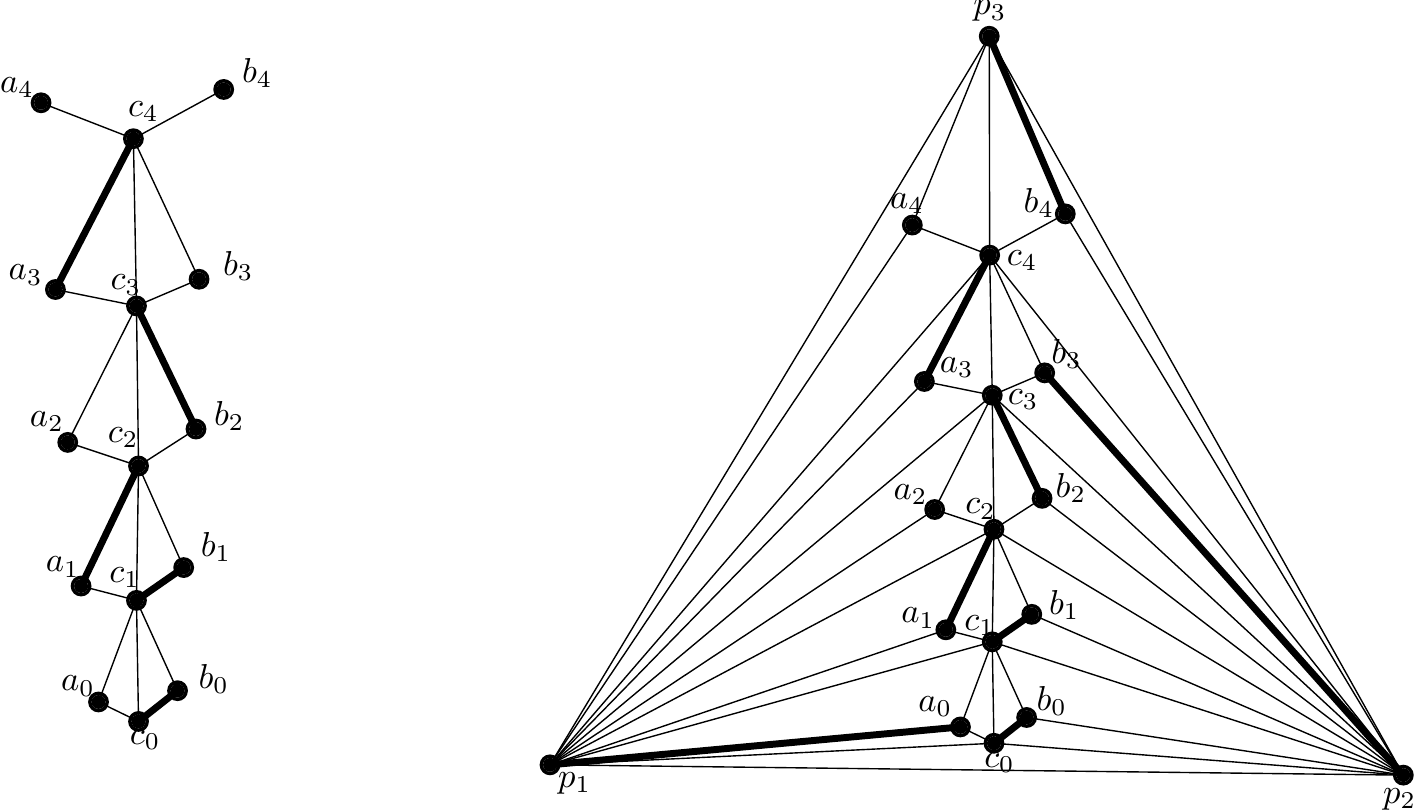}   
  \caption{(a) A point set $P$ in general position, where $G_\bigtriangledown(P)$ has a maximum matching of size $\lceil\frac{n-2}{3}\rceil$ \cite{Panahi}. (b) A point set $P$ in general position, where $G_\bigtriangledown(P)$ is $3$-connected and has a maximum matching of size $\left\lceil\frac{n+5}{3}\right\rceil$.}
\label{fig15and18vertex}
  \end{figure} 
\paragraph{Some graphs for which our bound is tight:}In Fig. \ref{fig15and18vertex} (a), a point set $P$ consisting of $15$ points and the corresponding graph $G_\bigtriangledown(P)$ is given. This graph has a maximum matching (shown in thick lines) of size $\left\lceil\frac{|P|-2}{3}\right\rceil =5$. This is the same example as given by Panahi et al. \cite{Panahi}. By adding more triplets of points $(a_i, b_i, c_i)$, $i>4$, into $P$, following the same pattern, we can show that for any $n\ge 15$ which is a multiple of $3$, there is a point set $P$ of $n$ points in general position, such that a maximum matching in $G_\bigtriangledown(P)$ is of cardinality $\left\lceil\frac{|P|-2}{3}\right\rceil$.

We can also show that, for any $n \ge 13$, which is one more than a multiple of three, there is a point set $P'$ on $n$ points in general position, such that a maximum matching in $G_\bigtriangledown(P')$ is of cardinality $\left \lceil\frac{|P'|-2}{3}\right\rceil$. For example, take the point set $P'=P \setminus \{a_0, b_0\}$ where $P$ is the point set of triplets described in the paragraph above. However, when $|P|$ is one less than a multiple of three, we do not have an example where our bound is tight. 

Thus, our bound is tight in all cases except when $n$ is one less than a multiple of three. From the examples above, it is clear that no bound better than $\left\lceil\frac{n-1}{3} \right\rceil$ is possible. It remains open whether our bound can be improved to $\left\lceil\frac{n-1}{3}\right\rceil$.
\subsection{A 3-connected down triangle graph without perfect matching}
The example given by Panahi et al. \cite{Panahi}, for a point set $P$ for which $G_\bigtriangledown(P)$ has a maximum matching of size $\left\lceil\frac{n-2}{3}\right\rceil$, contained many cut vertices. However, for general planar graphs, we get a better lower bound for the size of a maximum matching, when the connectivity of the graph increases. By Theorem \ref{thmnishi}, we know that any $3$-connected planar graph on $n$ vertices has a matching of size $\left\lceil\frac{n+4}{3}\right\rceil$, if $n \ge 14$ and has a matching of size $\left\lfloor\frac{n}{2}\right\rfloor$ if $n<14$ or it is 4-connected. Hence, it was interesting to see whether there exist a point set $P$ in general position, with an even number of points, such that $G_\bigtriangledown(P)$ is $3$-connected but does not contain a perfect matching. The answer is positive. 

Consider the graph given in Fig. \ref{fig15and18vertex} (b), which shows a point set $P$ of $18$ points in general position and the corresponding graph $G_\bigtriangledown(P)$. This graph has a maximum matching (shown in thick lines) of size $8$.  We can follow the pattern and go on adding points $a_i$, $b_i$ and $c_i$, for $i >4$ to the point set such that when $P= \{a_0, b_0, c_0, \ldots, a_k, b_k, c_k, p_1, p_2, p_3\}$, $G_\bigtriangledown(P)$ is a $3$-connected graph with a maximum matching of size $\left \lceil\frac{|P|+5}{3} \right\rceil$. It can be verified that $G_\bigtriangledown(P\setminus \{a_0\})$ and $G_\bigtriangledown(P\setminus \{a_0, b_0\})$ are also $3$-connected and their maximum matchings have size $\left\lceil\frac{|P|+5}{3}\right\rceil$. Thus, for $3$-connected down triangle graphs corresponding to point sets in general position, the best known lower bound for maximum matching is $\left\lceil\frac{n+4}{3}\right\rceil$ and the examples we discussed above show that it is not possible to improve the bound above $\left\lceil\frac{n+5}{3}\right\rceil$.  
\section{Some properties of $G_{\davidsstar}(P)$}
In this section, we will prove some properties of the graph $G_{\davidsstar}(P)$. 
\subsection{Connectivity}
For a point set $P$, it is easy to see that $G_{\davidsstar}(P)$ is connected because it is the union $G_{\bigtriangledown}(P)$ and $G_{\bigtriangleup}(P)$, which are themselves connected graphs by Lemma \ref{pathintriangle}.
\subsection{Number of degree-one vertices}
We will prove that $G_{\davidsstar}(P)$ can have at most two degree-one vertices.
\begin{lemma}\label{bidirectionNeighbour}
Let $P$ be a point set, $p\in P$, and $i \in \{1, 2, 3\}$. In $G_{\davidsstar}(P)$ the vertex $p$ has at least one neighbour in $\overline{V_i}(p)$, if $\overline{V_i}(p) \ne \emptyset$. Similarly, the vertex $p$ has at least one neighbour in ${V_i}(p)$, if ${V_i}(p) \ne \emptyset$. 
\end{lemma}
\begin{proof}
 If $\overline{V_i}(p)$ is non-empty, then by Property \ref{obs2}, in $G_{\bigtriangledown}(P)$ vertex $p$ has a neighbour in $\overline{V_i}(p)$. Similarly, we can prove that if ${V_i}(p)$ is non-empty, then in $G_{\bigtriangleup}(P)$, vertex $p$ has a neighbour in ${V_i}(p)$. Since $G_{\davidsstar}(P)$ = $G_{\bigtriangledown}(P) \cup G_{\bigtriangleup}(P)$, the proof is complete.
\qed
\end{proof}
\begin{lemma}
 For a point set $P$, its $G_{\davidsstar}(P)$ can have at most two degree-one vertices.
\end{lemma}
\begin{proof}
Let $P$ be a point set and $p\in P$ be such that the vertex $p$ is of degree one in $G_{\davidsstar}(P)$. From Lemma \ref{bidirectionNeighbour}, there exists an $i \in \{1, 2, 3\}$ such that exactly one of
 ${V_i}(p)$ and $\overline{V_i}(p)$ is non-empty and contains all points in $P \setminus \{p\}$. Without loss of generality, assume that $i=1$ and $V_1(p) \ne \emptyset$, $\overline{V_1}(p) = \emptyset$. Then, for $j \in \{2, 3\}, {V_j}(p)= \emptyset$ and $\overline{V_j}(p)=\emptyset$. 
 
 Let $q \in P$ be another point such that the vertex $q$ is of degree one in $G_{\davidsstar}(P)$. We know that $q \in V_1(p)$ and hence $p \in \overline{V_1}(q)$. Again, from Lemma \ref{bidirectionNeighbour}, we get $V_1(q) =\emptyset$ and for $j \in \{2, 3\}, {V_j}(q)= \emptyset$ and $\overline{V_j}(q)=\emptyset$. Thus, $P\setminus\{q\} \subseteq \overline{V_1}(q)$.
 
 If there is a third point $r \in P$ such that the vertex $r$ is also of degree one in $G_{\davidsstar}(P)$, then we get $r \in V_1(p)$ and $r \in \overline{V_1}(q)$. This will mean that $V_1(r) \ne \emptyset$ and $\overline{V_1}(r) \ne \emptyset$. By Lemma \ref{bidirectionNeighbour}, this is not possible because $r$ is of degree one. Thus we conclude that  $G_{\davidsstar}(P)$ has at most two degree-one vertices.
\qed
\end{proof}
 \subsection{Block cut point graph}
Let $G(V, E)$ be a graph. A block of $G$ is a maximal connected subgraph having no cut vertex. The block cut point graph of $G$ is a bipartite graph $B(G)$ whose vertices are cut-vertices of $G$ and blocks of $G$, with a cut-vertex $x$ adjacent to a block $X$ if $x$ is a vertex of block $X$. For a connected graph, the block-cutpoint graph is always a tree \cite{Diestel}. For a connected graph, its block cut point graph gives information about its $2$-connectivity structure. In this section, we will show that the block cut point graph of $G_{\davidsstar}(P)$ is a simple path.
\begin{lemma}\label{numcomp}
Let $P$ be a point set and $p \in P$ be a cut vertex of $G_{\davidsstar}(P)$. Then, there exists an $i \in \{1, 2, 3\}$ such that ${V_i}(p)\ne \emptyset$, $\overline{V_i}(p)\ne \emptyset$ and for all $j \in \{1, 2, 3\} \setminus\{i\}$, ${V_j}(p) = \emptyset$ and $\overline{V_j}(p) = \emptyset$. Moreover, $G_{\davidsstar}(P) \setminus p$ has exactly two connected components, one containing all vertices in $V_i(p)$ and the other containing all vertices of $\overline{V_i}(p)$.
 \end{lemma}
\begin{proof}
Since $p$ is a cut vertex of $G_{\davidsstar}(P)$, we know that there exist $v_1, v_2 \in P$ that are in different components of $G_{\davidsstar}(P) \setminus p$. We will show that $v_1$ and $v_2$ should be in opposite cones with reference to the apex point $p$.

Without loss of generality, assume that $v_1 \in A_1(p) \cap P \setminus \{p\}$. If $v_2 \in ( A_1(p) \cup A_2(p) \cup A_6(p)) \cap (P \setminus \{p\})$, then, $p \notin \bigtriangledown v_1 v_2$ and hence by Lemma \ref{pathintriangle}, there is a path in $G_{\bigtriangledown}(P)$ between $v_1$ and $v_2$ that does not pass through $p$, which is not possible. Similarly, if $v_2 \in (A_3(p) \cup A_5(p)) \cap (P \setminus \{p\})$, then, $p \notin \bigtriangleup v_1 v_2$ and there is a path in $G_{\bigtriangleup}(P)$ between $v_1$ and $v_2$ that does not pass through $p$, which is not possible. Therefore, $v_2 \in A_4(p)$, the cone which is opposite to $A_1(p)$ which contains $v_1$. Thus any two points $v_1$ and $v_2$ which are in different connected components of $G_{\davidsstar}(P) \setminus p$, are in opposite cones around $p$. 

Let $C_1$ and $C_2$ be two connected components of $G_{\davidsstar}(P) \setminus p$ with $v_1 \in C_1$ and $v_2 \in C_2$. Without loss of generality, assume that such $v_1 \in {V_1}(p)$ and $v_2 \in \overline{V_1}(p)$. From the paragraph above, we know that every vertex of $G_{\davidsstar}(P) \setminus p$ which is not in $C_1$ is in $\overline{V_1}(p)$ and every vertex of $G_{\davidsstar}(P) \setminus p$ which is not in $C_2$ is in ${V_1}(p)$. This implies that for all $j \in \{2, 3\}$, ${V_j}(p) = \emptyset$ and $\overline{V_j}(p) = \emptyset$. This proves the first part of our lemma. 

For any $v_1, v_2 \in \overline{V_i}(p)$, we have $p \notin \bigtriangledown v_1 v_2$ and hence by Lemma \ref{pathintriangle}, there is a path in $G_{\bigtriangledown}(P)$ between $v_1$ and $v_2$ that does not pass through $p$. Similarly, for any $v_1, v_2 \in V_i(p)$, $p \notin \bigtriangleup v_1 v_2$ and there is a path in $G_{\bigtriangleup}(P)$ between $v_1$ and $v_2$ that does not pass through $p$. Therefore, there are exactly two connected components in $G_{\davidsstar}(P) \setminus p$, one containing all vertices in $V_i(p)$ and the other containing all vertices of $\overline{V_i}(p)$.
\qed
\end{proof}
\begin{theorem}\label{blocks}
Let $P$ be a point set in general position and let $k$ be the number of blocks of $G_{\davidsstar}(P)$. Then, the blocks of $G_{\davidsstar}(P)$ can be arranged linearly as $B_1,B_2, \ldots B_k$ such that, for $i > j$, $B_i \cap B_{j}$ contains a single (cut) vertex $p_i$ when $j = i+1$ and $B_i \cap B_{j}$ is an empty graph otherwise. That is, the block cut point graph of $G_{\davidsstar}(P)$ is a path. 
\end{theorem}
\begin{proof}
If $G_{\davidsstar}(P)$ is two-connected, there is only a single block and the lemma is trivially true. 

Since $G_{\davidsstar}(P)$ is a connected graph, its block cut point graph is a tree. Any two blocks can have at most one vertex in common and the common vertex is a cut vertex. From Lemma \ref{numcomp}, we also know that three or more blocks cannot share a common (cut) vertex. If a block $B_i$ of $G_{\davidsstar}(P)$ is such that, in the block cut point graph of $G_{\davidsstar}(P)$, the node corresponding to block $B_i$ is a leaf node, $B_i$ is adjacent to only one another block and they share a single (cut) vertex. 

If the node corresponding to $B_i$ is not a leaf node of the block cut point graph, we know that $B_i$ shares (distinct) common vertices with at least two other blocks $B_{i'}$ and $B_{i''}$. Therefore, two vertices in $B_i$ are cut vertices of $G_{\davidsstar}(P)$. Let $v_1, v_2$ be these cut vertices. We will show that there cannot be a third such cut vertex in $B_i$.

By Lemma \ref{numcomp}, we know that $G_{\davidsstar}(P) \setminus v_1$ has exactly two components and since $B_i$ is $2$-connected initially, all vertices of $B_i$ except $v_1$ are in the same connected component of $G_{\davidsstar}(P) \setminus v_1$. By Lemma \ref{numcomp}, all vertices of $B_i$ lie in the same (designated) cone with apex $v_1$. Without loss of generality, assume that all vertices in $B_i \setminus v_1$ are in $V_1(v_1)$. In particular, $v_2 \in V_1(v_1)$ and hence $v_1 \in \overline{V_1}(v_2)$. Similarly, since $v_2$ is a cut vertex, all vertices of $B_i$ lie in the same (designated) cone with apex $v_2$. Since $v_1 \in \overline{V_1}(v_2)$, all vertices in $B_i \setminus v_2$ are in $\overline{V_1}(v_2)$. If $v_3$ is a vertex in $B_i$, distinct from $v_1$ and $v_2$, then from the discussion above, we get $v_3 \in V_1(v_1)$ and $v_3 \in \overline {V_1}(v_2)$. Hence $v_1 \in \overline{V_1}(v_3)$ and $v_2 \in V_1(v_3)$. Suppose $v_3$ is a cut vertex in $G_{\davidsstar}(P)$. Since $v_1$ and 
$v_2$ are in the same connected component of $G_{\davidsstar}(P) \setminus v_3$, it is a contradiction to Lemma \ref{numcomp}, that $v_1 \in \overline{V_1}(v_3)$ and $v_2 \in {V_1}(v_3)$. 

Thus, if the node corresponding to $B_i$ is not a leaf node of the block cut point graph of $G_{\davidsstar}(P)$, then exactly two vertices in $B_i$ are cut vertices of $G_{\davidsstar}(P)$. Since no three blocks can share a common vertex by Lemma \ref{numcomp}, we are done.
\qed
\end{proof}
\subsection{Number of Edges of $G_{\davidsstar}(P)$}
Since $G_\bigtriangledown(P)$ and $G_\bigtriangleup(P)$ are planar graphs and $G_{\davidsstar}(P)=G_\bigtriangledown(P) \cup G_\bigtriangleup(P)$, it is obvious that $G_{\davidsstar}(P)$ has at most $2 \times (3n-6) = 6n -12$ edges, where $n=|P|$ \cite{Diestel}. In this section, we show that for any point set $P$, its $G_{\davidsstar}(P)$ has a spanning tree of a special structure, which will imply that $G_{\davidsstar}(P)$ can have at most $5n-11$ edges.
\begin{lemma}\label{spanningTree}
 For a point set $P$, the intersection of $G_\bigtriangledown(P)$ and $G_\bigtriangleup(P)$ is a connected graph.
\end{lemma}
\begin{proof}
 We will prove this algorithmically. At any point of execution of this algorithm, we maintain a partition of $P$ into two sets $S$ and $P \setminus S$ such that the induced subgraph of $G_\bigtriangledown(P) \cap G_\bigtriangleup(P)$ on $S$ is connected. When the algorithm terminates, we will have $S=P$, which will prove the lemma.

 We start by adding any arbitrary point $p_1 \in P$ to $S$. The induced subgraph of $G_\bigtriangledown(P) \cap G_\bigtriangleup(P)$ on $S$ is trivially connected now. 
 
 At any intermediate step of the algorithm, let $S = \{p_1, p_2, \ldots, p_k \} \ne P$, such that the invariant is true. We will show that we can add a point $p_{k+1}$ from $P \setminus S$ into $S$, and still maintain the invariant. 

 For any point $p\in S$, let $d_1(p) = \displaystyle\min_{i \in \{1, 2, 3\}, p' \in V_i(p) \cap P \setminus S}c_i(p, p')$, $d_2(p) = \displaystyle\min_{i \in \{1, 2, 3\}, p' \in \overline{V_i}(p)\cap P \setminus S}\overline{c_i}(p, p')$ and $d(p)=\min(d_1(p), d_2(p))$. Since $|P\setminus S|\ge 1$, $d(p) < \infty$. Let $d=\displaystyle \min_{p \in S} d(p)$. 

Consider $p \in S$ such that $d(p)= d$. By definition of $d$, such a point exists. Consider the area enclosed by the hexagon around $p$ which is defined by $H_p=\displaystyle\bigcup_{i=1}^3 \{p' \in C_i(p) \mid c_i(p, p') \le d\} \cup \displaystyle\bigcup_{i=1}^3 \{p' \in \overline{C_i}(p)\mid \overline{c_i}(p, p') \le d\}$. (See Fig. \ref{Fighexagon} (a).) We know that there exists a point $q \in P \setminus S$ such that $q$ is on the boundary of $H_p$. We claim that $pq$ is an edge in $G_\bigtriangledown(P) \cap G_\bigtriangleup(P)$. 

\begin{figure}[h]
  \centering
  \includegraphics[scale=0.6]{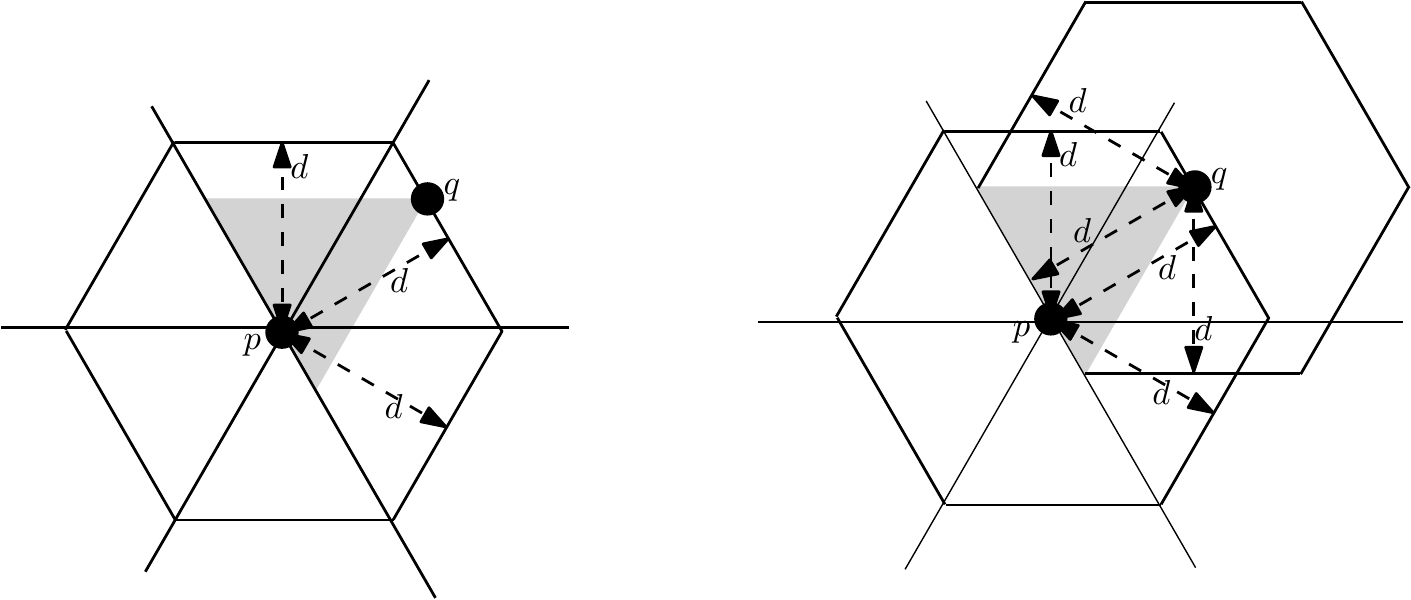}   
  \caption{(a) Closest point to $p$. (b) Hexagons around closest pairs.}
\label{Fighexagon}
  \end{figure} 

Let $H_q=\displaystyle\bigcup_{i=1}^3 \{p' \in C_i(q) \mid c_i(q, p') \le d\} \cup \displaystyle\bigcup_{i=1}^3 \{p' \in \overline{C_i}(q)\mid \overline{c_i}(q, p') \le d\}$, which is a hexagonal area around $q$. (See Fig. \ref{Fighexagon} (b).) Without loss of generality, assume that $q \in C_1(p)$. Note that, by Property \ref{obs1}, $c_1(p, q)=\overline{c_1}(q, p)=d$ and hence, $\bigtriangledown pq \cup \bigtriangleup pq \subseteq H_p \cap H_q$. 

If there exists a point $q' \in (P \setminus \{q\}) \setminus S$ such that $q'$ lies in the interior of $H_p$, then $d(p)<d$, which is a contradiction. Similarly, if there exists a point $p' \in (P \setminus \{p\}) \cap S$ such that $p'$ lies in the interior of $H_q$, then $d(p)<d$. This is also a contradiction. Therefore, $H_p \cap H_q \cap (P \setminus \{p, q\}) = \emptyset$. Since, $\bigtriangledown pq \cup \bigtriangleup pq \subseteq H_p \cap H_q$, this implies that $\bigtriangledown pq \cap (P \setminus \{p, q\}) = \emptyset$ and $\bigtriangleup pq \cap (P \setminus \{p, q\}) = \emptyset$. This implies that $pq$ is an edge in $G_\bigtriangledown(P)$ as well as in $G_\bigtriangleup(P)$.

Since $pq$ is an edge in $G_\bigtriangledown(P) \cap G_\bigtriangleup(P)$, we can add $p_{k+1}=q$ to the set $S$, thus increasing the cardinality of $S$ by one, and still maintaining the invariant that the induced subgraph of $G_\bigtriangledown(P) \cap G_\bigtriangleup(P)$ on $S$ is connected. Since we can keep on doing this until $S=P$, we conclude that $G_\bigtriangledown(P) \cap G_\bigtriangleup(P)$ is connected. 
\qed
\end{proof}
\begin{theorem}
  For a set $P$ of $n$ points in general position, $G_{\davidsstar}(P)$ has at most $5n-11$ edges and hence its average degree is less than $10$.
\end{theorem}
\begin{proof}
Since $G_{\bigtriangledown}(P)$ and $G_{\bigtriangleup}(P)$ are both planar graphs we know that each of them can have at most $3 n-6$ edges. From Lemma \ref{spanningTree}, we know that the intersection of $G_{\bigtriangledown}(P)$ and $G_{\bigtriangleup}(P)$ contains a spanning tree and hence they have at least $n-1$ edges in common. From this, we conclude that the number of edges in $G_{\davidsstar}(P) = G_\bigtriangledown(P) \cup G_\bigtriangleup(P)$ is at most $(3n-6) +(3n-6) -(n-1) = 5n-11$. Hence,the average degree of $G_{\davidsstar}(P)$ is less than $10$.
\qed
\end{proof}
\begin{corollary}
 For a set $P$ of $n$ points in general position, its $\Theta_6$ graph has at most $5n-11$ edges.
\end{corollary}
\section{Conclusions}
We have shown that for any set $P$ of $n$ points in general position, any maximum $\bigtriangledown$ (resp. $\bigtriangleup$) matching of $P$ will match at least $2\left(\left\lceil\frac{|P|-2}{3} \right\rceil\right)$ points. This also implies that any half-$\Theta_6$ graph for point sets in general position has a matching of size at least $\left\lceil\frac{|P|-2}{3} \right\rceil$. This bound is tight except when $|P|$ is one less than a multiple of three. We also proved that when $P$ is in general position, the block cut point graph of its $\Theta_6$ graph is a simple path and that the $\Theta_6$ graph has at most $5n-11$ edges. It is an interesting question to see whether for every point set in general position, its $\Theta_6$ graph contains a matching of size $\left\lfloor \frac{|P|}{2} \right\rfloor$. So far, we were not able to get any counter examples for this claim and hence we conjecture the following.
\begin{conjecture}
 For every set of $n$ points in general position, its $\Theta_6$ graph contains a matching of size $\left\lfloor \frac{n}{2} \right\rfloor$.
\end{conjecture}


\bibliographystyle{splncs03}
\bibliography{mybib}
\end{document}